\documentclass[12pt]{article}

\usepackage[utf8]{inputenc}
\usepackage{amsfonts,amsthm,amsmath}
\usepackage{amssymb} 
\usepackage{url}
\usepackage[english]{babel}
\usepackage[small,nohug,heads=vee]{diagrams} 
\diagramstyle[labelstyle=\textstyle]

\title{
	Conformal reference frames for Lorentzian manifolds
}
\author{Innocenti V. Maresin}

\ifx\C\undefined\else\makeatletter\let\C\@undefined\makeatother\fi
\newcommand{\C}{{\mathbb{C}}}
\ifx\R\undefined\newcommand{\R}{{\mathbb{R}}}\fi

\newcommand{\mul}{\hspace{0.125em}}
\newcommand{\abs}[1]{{\left|{#1}\right|}}
\newcommand{\noZero}{\setminus\{0\}}

\newcommand{\Mink}{{\mathbb{M}}}
\newcommand{\Nfin}{\mathbf{N}_{\mathrm{a}}} 
\newcommand{\Md}{{\Mink^{1+d}}}

\newcommand{\X}{X}
\newcommand{\sky}[1]{{\mathfrak{S}_{{#1}}}}
\let\scoord\varsigma 
\newcommand{\Csky}[1]{{\mathbb{S}^*_{{{#1}}1,0}}} 
\newcommand{\hs}{\Gamma_{\mathrm{hol}}} 
\newcommand{\twId}[1]{{{\scriptstyle [}{{#1}}{\scriptstyle ]}}}
\let\fib F
\newcommand{\fibs}[1]{F_{{#1}}}
\newcommand{\skyreg}[1]{\Omega_{{#1}}}
\newcommand{\NM}{N_{M}}
\newcommand{\mfib}{j}\let\toM\mfib
\newcommand{\msky}[1]{{\mathfrak{M}_{{#1}}}}

\newcommand{\xfib}{{\mathrm x}}
\newcommand{\bPm}{p} 
\let\Th N 
\newcommand{\Tv}{T_{\mathrm v}}

\newcommand{\CPone}{{{\C}\mathbf{P}^1}}
\newcommand{\cW}[1]{{\mathcal{O}_{{#1}}(1,0)}} 
\newcommand{\spinorsSections}[1]{{\mathbb{S}^{1,0}_{{#1}}}} 
\newcommand{\ConjcW}[1]{{\mathcal{O}_{{#1}}(0,1)}}
\newcommand{\ConjspinorsSections}[1]{{\mathbb{S}^{0,1}_{{#1}}}} 
\let\HB\scoord
\newcommand{\Rfuture}[1]{{\mathcal{L}^+_{{#1}}}}
\newcommand{\RR}[1]{{\mathcal{L}^\R_{{#1}}}}
\newcommand{\RC}[1]{{\mathcal{O}_{{#1}}(1,1)}}
\newcommand{\sM}{\operatorname{s_\Mink}}
\newcommand{\s}{\operatorname{s}}

\let\contact\vartheta
\let\contactfib\vartheta

\let\Conj\bar
\let\ConjHB\pi
\newcommand{\Conjsky}[1]{{{\Conj{\mathfrak{S}}}_{{#1}}}}
\newcommand{\CConjsky}[1]{{\mathbb{S}^*_{{{#1}}0,1}}} 

\newcommand{\oneHalf}{{^1\!/\!_2}}\theoremstyle{plain}

 \newtheorem{theorem}{Theorem}
 \newtheorem{proposition}{Proposition}
\theoremstyle{definition}
 \newtheorem*{definition}{Definition}
 \newtheorem*{example}{Example}
 \newtheorem*{acks}{Acknowledgements}
\theoremstyle{remark}
 \newtheorem*{corollary}{Corollary}
 \newtheorem*{remark}{Remark}

\begin{document}
\maketitle
\begin{abstract}
The definition of a conformal reference frame is given, that is,
of a special projection of the six-dimensional skies bundle of a Lorentzian manifold
(or five-dimensional twistor space) 
to a three-dimensional manifold.
An example is constructed—conformal compactification—for the Minkowski space.
The {\em celestial transform} of Lorentzian vectors is defined, a kind of spinor correspondence, 
based on the complex structure on the skies.
An 1-form generating the contact structure in the twistor space (when the latter is smooth)
is expressed explicitly as a (line bundle)-valued form.
A theorem is proved on the projection of the said 1-form
to the fiberwise normal bundle of a reference frame.
It entails the {\em flow of time equation} that expresses 
the space-time derivative of sky images in terms of the celestial transform of 4-vectors.

The Appendix is less mathematical than the main body 
and discusses the causal relation 
in context of the FLRW cosmology and its “natural” conformal reference frame. 
\end{abstract} 

\begin{acks}
This research is supported in part by the Russian Foundation for Basic Research,
(grant 16-01-00117~A “Complex Problems in Mathematical Physics”)
and the Russian Academy of Sciences 
(programme I.37$\Pi$ “Non-linear Dynamics in the Mathematical and Physical Sciences”, 
project 0014-2015-0037 “Complex and algebro-geometric methods in problems of nonlinear dynamics”). 
The author thanks A.~G. Sergeev for the help in preparing the work for publication, 
and Bill Everett for some suggestions on the English text.
\end{acks}

This paper (except the Appendix) is about to be published in 
{\it Theoretical and Mathematical Physics}, {\bf 191}(2): 682–691 (2017), 
DOI: 10.1134/S0040577917050099. 
{\small The journal version has minor differences 
due to the redaction’s preferences for vocabulary, spelling, notation, punctuation, and grammar; 
also due to Russian–English text correspondence policy. 
This version follows the author’s preferences and also bears some later fixes.}

\subsection{Introduction}
Lorentzian manifolds are the standard framework for the space-time in physics, and there are some mathematical techniques to explore their geometry. 
The twistor approach to Lorentzian manifolds is focused 
on the five-dimensional space~$\mathfrak{N}$ of null geodesics. 
In the case of Minkowski space, that space can be embedded in a complex projective space, 
the {\em twistor space}, 
but in general (curved) case any canonical complex structure on twistors isn’t possible. 
Wherever $\mathfrak{N}$ could be defined,
for each point $x\in\X$ its sky $\sky x$ is embedded (or immersed) into~$\mathfrak{N}$, 
and $\mathfrak{N}$ is endowed with the natural contact structure 
\cite{LowConjecture}. 
A sky $\sky x$ has natural conformal structure and is isomorphic to $\CPone$. 
This allows formulating the well-known twistor correspondence 
in a geometrized way agreeable to the curved case; 
it’s free from aforementioned drawback~– 
inability to use any complex structure globally.

A conformal reference frame, defined in this paper,
gives a description of (four-dimensional) Lorentzian manifold~$\X$ 
in terms of immersion (sometimes partial) of its skies into a three-dimensional manifold. 
It is known that in the case of a globally hyperbolic~$\X$ 
its $\mathfrak{N}$ identifies with the spherical cotangent bundle over a Cauchy surface~$M\subset\X$. Generalization to certain weaker conditions on~$\X$ are possible as well.
Then, instead of spheres—submanifolds of $\mathfrak{N}$—we can consider
the image of~$\sky x$ in~$M$, an immersion in non-singular case.
Conformal reference frames generalize Cauchy surfaces in some way. 
Ignoring the technical difficulties of differential geometry on~$\mathfrak{N}$, 
conformal reference frame should be understood as 
a smooth mapping of~$\mathfrak{N}$, or part thereof, 
to an arbitrary 3-manifold $M$, the mapping compatible with the contact structure 
and having certain non-degeneracy condition 
(i.e., for every~$x\in\X$ the mapping $ T\sky x \to TM$ has rank~2 
on a non-empty open subset of~$\sky x$). 
In this paper we present the {\em flow of time equation} 
expressing the dependence of the sky image $\sky x$ on the point~$x$, 
and how it is related to the contact structure and causality, 
from the differential rather than the topological standpoint. 
It is shown that holomorphic sections of the line bundle~$\cW{}$ over the sky 
should be considered spinors, 
and the “bundle of sizes”~$\RR{}$ (with the weight-$(\oneHalf, \oneHalf)$ representation of~$\mathrm{SL}(2,\C)$) 
is the natural range of the contact form on the twistor space~$\mathfrak{N}$. 

\section{The bundle of skies}
\subsection{Preliminaries}
This subsection sets out the facts known about Lorentzian manifolds.
Usually, they are understood as pseudo-Riemann manifolds of signature~(1,3)
i.e. $({+}{-}{-}{-})$\footnote{
	Many authors use a $({-}{+}{+}{+})$ metric.
	The difference may affect algebraic aspects of the theory, but the geometry remains the same.
}
This paper requires an additional structure, namely:
\begin{definition}
A {\em Lorentzian manifold} is a pseudo-Riemann four-dimensional manifold~$\X$ 
with the metric $g$ of the signature $({+}{-}{-}{-})$ 
and the {\em time orientation} at each point $x\in\X $ 
(i.e. one of the two connected components of the cone $\{\, v \in T_x\X \ | \ g(v) > 0\,\} $ is chosen 
as “chronological future”), in a continuous fashion.
\end{definition}
Manifolds~$\X$ satisfying this definition are referred to as {\em space-times} in \cite{LowConjecture}.
All manifolds are assumed to be smooth ($C^\infty$).

\begin{definition}
For each $x\in\X$ the boundary of its “chronological future” cone in~$T_x\X$ 
is called {\em the future light cone} and is denoted by $\mathcal{C}^+_x$.
\end{definition}
\begin{definition}
The sky $\sky x$ is the base of the cone~$\mathcal{C}^+_x$ in~$T_x\X$, 
and elements of the former will be denoted $\mathrm{P}v$, where $v\in\mathcal{C}^+_x\noZero$.
The disjoint union of all skies (of all points of $\X$) forms 
a smooth locally trivial bundle over~$\X$, 
denoted by~$\sky{}\X$, 
with the projection map $\xfib: \sky{}\X \to \X$.
\end{definition}
\begin{remark}
An element of~$\sky x$ is virtually a null direction at~$x$,
(1-subspace in~$T_x\X$),
whereas $\mathrm{P}: v\mapsto\{\,\lambda\mul v\ |\ \lambda\in\C\,\}$
means projectivization.
\end{remark}
\begin{definition}
Let $\Tv\sky{}\X$ denote the disjoint union of all tangent bundles $T\sky x$ for all~$x\in\X$. 
In other words: the vertical subbundle~$\ker d\xfib$ in~$T(\sky{}\X)$.
\end{definition}
\begin{definition}
The {\em geodesic flow} $\fibs{}\X$ is a distribution of 1-subspaces 
in the tangent bundle $T(\sky{}\X)$ of the total space of the bundle~$\sky{}\X$, 
defined by the equations $dx \parallel v $ (the differential of $x$ is collinear to~$v$), 
$\nabla v = 0$ ($v$ is constant w.r.t. the Levi-Civita connection), 
where $v\in\mathcal{C}^+_x \noZero$ is a vector representing given point of the sky. 
\end{definition}

\begin{remark}
Integrating the flow~$\fibs{}\X$ gives 
the “light” foliation of the total space $\sky{}\X$. 
Its leafs represent null\footnote{
	{\it TMF} replaced the word “null” with “isotropic” here and
	in some further instances.
} (light-like) geodesic curves on~$\X$,
maximally extended in both time directions and lifted to~$\sky{}\X$ naturally, 
by the mapping~$\mathrm{P}$ of tangent vectors to respective skies.
\end{remark}
\begin{definition}
The equivalence relation $\twId {(x_1, \mathrm{P}v_1)} = \twId{(x_2, \mathrm{P}v_2)}$ 
on the total space of $\sky{}\X$ 
is a property of null vectors $v_1\in\mathcal{C}^+_{x_1} \noZero$ and $v_2\in\mathcal{C}^+_{x_2} \noZero$ 
to lie on the same null geodesic, 
i.e. a leaf of the foliation; see previous remark.
The null vectors represent certain null directions.
\end{definition}

From \cite{LeBrun} and \cite{LowConjecture} we know that $\mathfrak{N}$,
that is defined as the quotient space of $\sky{}\X $ by 
the equivalence relation introduced above, 
possesses a natural contact structure (when $\mathfrak{N}$ is smooth). 
In this paper a weaker version of this statement will be used, 
adapted to the fact that $\mathfrak{N}$ does not always admit a manifold structure.
\begin{definition}
Let’s choose
for each point $w\in\sky{}\X$ its representative $v\in\mathcal{C}^+_x \noZero$, 
where $x : = \xfib(w)$, in a smooth fashion.
Then $\contactfib = (v.dx)$ is an 1-form on the total space $\sky{}\X$. 
\end{definition}
\begin{remark}
Obviously, the 1-form defined so is meaningful only up to multiplication by a positive function. 
However, there is a real line bundle defined {\em globally} over~$\sky{}\X$
that is the range of the form~$\contactfib$ invariantly defined. 
Moreover, that line bundle is oriented (that is, has the positive side marked).
This construction will be postponed until 2.2.
\end{remark}
\begin{remark}
The 1-form $\contactfib$ is smooth and never vanishes. 
But $\contactfib$ nullifies all the $\fibs{}\X$ (because 
every Lorentzian null direction is orthogonal to itself) 
and $\Tv\sky{}\X$ (because $d\xfib = 0$).
In other words, the five-dimensional bundle $T(\sky{}\X)\ /\ \fibs{}\X$ has 
a continuous distribution of homogeneous co-oriented hyperplanes~$\contactfib = 0$ 
containing $T{\sky x}$ for all~$x\in\X$. 
\end{remark}

\subsection{Defining a conformal reference frame}
\begin{definition}
A conformal reference frame~$(\skyreg{}, M, \toM)$ of the Lorentzian manifold~$\X$ is defined as:
\begin{itemize}
\item such open subset $\skyreg{} \subset \sky{}\X$ that 
each its fiber $\skyreg x := \skyreg{} \cap \sky x,\ x\in\X$ is not empty;
\item a 3-dimensional smooth manifold~$M$;
\item such smooth map~$\mfib$ of $\skyreg{}$ to~$M$ that satisfies following conditions:
$$
\text{(f)}\quad\quad \forall w_1,w_2\in\skyreg{}\quad
\twId{w_1}=\twId{w_2}\Rightarrow \mfib(w_1)=\mfib(w_2)
$$
(i.e. the map is constant along leafs of the “light” foliation),
$$
\text{(c)}\quad\quad\quad\quad\quad\forall w\in\skyreg{}\quad
\contactfib|_w \in \mfib^*(T^*_{\mfib(w)}M)\quad
$$
(i.e. the oriented distribution of 1-subspaces in $T^*\skyreg{}$, represented by the form $\contactfib$, belongs everywhere to~$\mfib^*(T^*M)$, the inverse image of the cotangent bundle),
and also\\
\strut~\hspace{4em}(d) \ \ the derivative of~$\mfib$ along $T\skyreg x$ (fibers of~$\skyreg{}$) is not degenerate
(i.e. the pushforward $\mfib_*: \Tv\skyreg{} \to TM$ of vertical tangent vectors has rank~2 everywhere).
\end{itemize}
\end{definition}
The image of $\skyreg x$ by~$\mfib$—a surface immersed
to $M$—will be denoted by $\msky x$ and called a {\em sky image}.
\begin{example}
Let $M$ be a Cauchy surface in a globally hyperbolic Lorentzian manifold. 
Denote by $\X$ the part of said manifold that lies after $M$, \footnote{
	With respect to the Big Bang cosmology (Robertson–Walker spaces), 
	$M$ can be, more generally, a {\em conformal} boundary of~$\X$;
	see~\cite{Hawking} 6.8 and 10.
}
and let $\toM$ be the projection (along null geodesics) 
to the {\em twistor bundle} $ST^*M$ (see~\cite{LowConjecture}) 
followed by projection of the latter onto~$M$. 
Also, let $\skyreg{}$ be the non-singular subset of $\sky{}\X$ with respect to 
mapping of the skies, 
i.e. where the mapping $\mfib_*: \Tv\sky{} \to TM$ has rank~2. 
If everywhere on~$\X$ $\skyreg x$ isn’t empty, 
then $(\skyreg{}, M, \toM)$ is a conformal reference frame for the manifold~$\X$.
\end{example}

\subsection{The derivative of sky images}
In this subsection we intend to define the derivative 
of the sky image $\msky x := j(\skyreg x)$ for $x\in\X$ with respect to~$x$.
We regard 
the homomorphism $dj: T\skyreg{} \to j^*(TM)$ of bundles over $\skyreg{}$ 
from the tangent bundle of~$\skyreg{}$ 
to the inverse image of the tangent bundle of~$M$ 
as the differential of~$j$. 
It must be distinguished from the pushforward homomorphism.
\begin{definition}
For a conformal reference frame, we call
$$
\NM\skyreg{} := j^*(TM)\ /\ dj(\Tv\skyreg{})
$$
the {\em fiberwise normal bundle}, 
where $\Tv\skyreg{} := \Tv\sky{}\X|_{\skyreg{}}\,$.
\end{definition}
\begin{remark}
Restriction of the fiberwise normal bundle $\NM\skyreg{}$ to~$\skyreg x$, 
i.e. its “fiber” $\NM\skyreg x = j|_{\skyreg x}^*(TM)\ /\ dj(T\skyreg x)$ is,
away of self-intersections,
the normal bundle~$N\msky x$ of the surface $\msky x$ in~$M$ 
pulled to~$\skyreg x$ by $ j: \skyreg{} \to M$.
\end{remark}
\begin{remark}
The concept of the normal bundle does not require any kind of structure but differentiability; 
it’s merely a quotient space.
\end{remark}
\begin{definition}
Given a local trivialization of~$\xfib: \skyreg{}\to\X$ (that is, such domain~$U\subset\X$ 
that $U\times\Delta$ is a subdomain of~$\skyreg{}$, where $\Delta$ is the standard disk along the fibers and $\xfib$ becomes a projection of~$U\times\Delta$ onto~$U$) 
and a mapping to~$M$ specified as $j: U\times\Delta \to M$,
the {\em derivative} of the family~$\msky x$ with respect to~$x\in\X$ is defined (locally) as the
homomorphism of vector bundles from~$\Delta\times T\X$ to the corresponding piece of~$\NM\skyreg{}$:
$$
d\msky x 
:= \frac{\partial j}{\partial x} / 
dj(T\Delta)
$$
where “/” denotes the quotient (by the image of $T\Delta$ under $j$). 
In other words, the vectors $\frac{\partial j}{\partial x}(\xi)\in j^*(TM)$ 
for all $(x,z)\in U\times\Delta,\ \xi\in T_x\X$ 
represent $d\msky x (\xi)\in\NM\skyreg{}$ (at the same $(x,z)$).
\end{definition}
The local trivialization exists everywhere and 
the value of the derivative~$d\msky x$ is obviously independent of its choice.

\noindent\parbox[b]{240pt}{
\begin{definition}
The bundle $\Th\skyreg{} := T\skyreg{} / \Tv\skyreg{}$ is a natural (independent of trivialization) 
domain of the mapping~$d\msky x$; see the diagram. 
Similarly, for all $\sky{}\X$
we denote by $\Th\sky{}\X := \xfib^*(T\X) = T(\sky{}\X)\ /\ \Tv\sky{}\X$ 
the tangent bundle to~$\X$ pulled to~$\sky{}\X$. 
Also, $d\xfib$ always denotes the projection of $T$ onto respective $\Th$. 
\end{definition}
}\hspace{24pt}
\parbox[b]{100pt}{\begin{diagram}
\Tv\skyreg{}	&\rTo	&{dj(\Tv\skyreg{})} \\
   \dInto	&	&\dInto \\
T\skyreg{}	&\rTo^{dj} &{j^*(T M)} \\
\dOnto^{d\xfib}	&\rdTo^{\bPm} &\dOnto \\
\Th\skyreg{}	&\rTo^{d\msky x} &\NM\skyreg{}
\end{diagram}
}
\begin{definition}
Under the same conditions we define 
$\displaystyle \bPm := dj\ / \ dj(\Tv\skyreg{}): T\skyreg{} \to \NM\skyreg{}$~– 
the tautological projection of~$T\skyreg{}$ onto the fiberwise normal bundle.
\end{definition}

\section{Example: Minkowski space}
\subsection{Definition and the spinor correspondence}
\begin{definition}
The {\em algebraic spinor correspondence} or the {\em Pauli transform} maps 
a vector from~$\R^4$, defined by the components $(x^0, x^1, x^2, x^3)$, 
to the Hermitian $2\times 2$ matrix 
$$
x^{AA'} = \frac{1}{2}\mul\begin{pmatrix}
x^0 + x^3	&& x^1 + i\mul x^2 \\
x^1 - i\mul x^2	&& x^0 - x^3
\end{pmatrix}.\quad\footnote{
	The real factor in the formula is more often chosen as $1/\sqrt{2}$ or $1$.
	Its value hasn’t importance for this paper.
}
$$
\end{definition}
\begin{remark}
It is easy to verify that every null 
(with respect to the $\eta = (dx^0)^2 - (dx^1)^2 - (dx^2)^2 - (dx^3)^2$ metric) 
vector $x\in\R^4$ under the condition $x^0\ge 0$ 
can be expressed as
$$
x^{AA'} = \psi^A \mul {\Conj\psi}^{A'},\ \psi\in\spinorsSections{},\quad
\text{where~}\spinorsSections{} := \C^2
$$ 
and, conversely, any $\psi\in\spinorsSections{}$ thus gives a vector from~$\mathcal{C}^+$ in this way. 
\end{remark}
\begin{definition}
The {\em Minkowski space} $\Mink$ is the~$\R^4$ coordinate space 
endowed with the standard pseudo-Euclidean metric~$\eta$ specified above.
\end{definition}
\begin{remark}
Here the action of the group~$\mathrm{SO}^+(1,3)$ on~$\Mink$ induces 
a representation of the group~$\mathrm{SL}(2,\C)$ on $\mathrm{Herm}(2)$ defined by 
$$
\mathrm{SL}(2,\C)\ni C\quad:
\quad x^{AA'} \longmapsto C^A_L \mul x^{LL'} \mul \Conj C_{L'}^{A'}
$$
(or, as matrix multiplication, $x^{AA'}$ is multiplied by~$C$ on the left and by~$C^*$ on the right); 
that is the tensor product of two 2-dimensional representations 
on the spaces $\spinorsSections{}$ and $\ConjspinorsSections{}$ (see below) respectively. 
\end{remark}
\begin{definition}
The dual space~$\Csky{}$ of spinors is the linear complex space 
dual to~$\spinorsSections{}$, 
and it’s equipped with 
a representation of the group $\mathrm{SL}(2,\C)$
that multiplies row vectors by the inverse matrix~$C^{-1}$ on the right.
\end{definition}
\begin{remark}
All the bundle $\sky{}\Mink$ has a natural trivialization, 
arising from the trivialization of $T\Mink$ by translations. 
Correspondence between projectivizations on the cone~$\mathcal{C}^+$ and 
the space~$\spinorsSections{}$ also enables
identification of~$\sky{}$—the base of the cone—with the Riemann sphere $\CPone = \mathbf{P}\spinorsSections{}$. 
On the other hand, each $\mathrm{P}\psi,\ \psi\in\spinorsSections{} \noZero$ corresponds 
to a homogeneous line $\{\,\scoord\in\Csky{}\ |\ \scoord \mul \psi = 0\,\}$, 
hence the projectivizations of the spaces $\Csky{}$ and $\spinorsSections{}$ are canonically isomorphic. 
\end{remark}
For the reasons explained below, it is convenient to represent a point of a sky in the Minkowski space 
as “$\mathrm{P}\scoord$”, where $\scoord\in\Csky{} \noZero$, 
and to identify it with a complex homogeneous line in~$\Csky{}$. 

\subsection{Complex line bundles}
\begin{definition}
The bundle $\mathcal{O}(k,l),\ k,l\in\mathbb{Z}$ over $\sky{}$ is a complex line bundle 
whose fiber at a point $\mathrm{P}\scoord,\ \scoord\in\Csky{} \noZero$ is a 1-dimensional space 
of all homogeneous bidegree~$(k,l)$ complex-valued functions 
on the homogeneous line~$\{\,\lambda \mul \scoord\ |\ \lambda\in\C\,\}\subset\Csky{}$.
The term “homogeneous bidegree” refers to the property
$$
\forall f\in\mathcal{O}(k,l)_{\mathrm{P}\scoord}\ \forall\lambda\in\C:
f(\lambda\mul\scoord) = \lambda^k\mul\Conj\lambda^l\mul f(\scoord).
$$
\end{definition}
\begin{remark}
Every homogeneous bidegree $(k,l)$ function on the whole $\Csky{}$ 
defines a section of~$\mathcal{O}(k,l)$ by its restriction to all complex homogeneous lines.
\end{remark}
\begin{remark}
In particular, $\cW{}$ is a holomorphic bundle, 
and every $\omega\in\spinorsSections{}$ defines its holomorphic section 
as a linear functional on~$\Csky{}$. 
The space of spinors $\spinorsSections{} = \hs(\cW{})$ 
is thus identified with the space of holomorphic sections.
\end{remark}
\begin{definition}
$\ConjspinorsSections{} := \overline{\spinorsSections{}} = \Gamma_{\mathrm{antihol}}(\ConjcW{})$ 
is the complex conjugate space of spinors.
Also, it has the dual space $\CConjsky{}$.
\end{definition}
\begin{remark}
As can be seen from the representation, $2\times 2$ Pauli matrices should be understood 
as elements of the tensor product $\spinorsSections{} \otimes \ConjspinorsSections{}$, 
so they specify sections of the bundle~$\RC{}$.
\end{remark}
\begin{definition}
The {\em celestial transform} $\sM: \Mink \to \Gamma(\RC{})$ 
for Lorentzian vectors denotes the same map 
from $\Mink$ to $\spinorsSections{} \otimes \ConjspinorsSections{}$ as for algebraic one, 
but its values are interpreted as (1,1)-homogeneous 
nonholomorphic polynomials $\scoord_A \mul x^{AA'} \mul \Conj\scoord_{A'}$ on $\Csky{}$ 
or, equivalently, as sections $\HB_A \cdot x^{AA'} \cdot \Conj\HB_{A'}$ 
of the nonholomorphic bundle~$\RC{} = \cW{} \otimes \ConjcW{}$ over~$\sky{}$. 
\end{definition}
We may tell the celestial transform the {\em geometric} spinor correspondence.
Moreover, $\sM$ will also denote the respective mapping 
from $\Mink \times \sky{} = \sky{}\Mink$ to the total space of~$\RC{}$.
\begin{remark}
Two bundles $\RC{}$ constructed over $\sky{}$ and $\Conjsky{} = \mathbf{P}\CConjsky{}$ do not differ 
except in the complex structure on their bases. 
They are canonically isomorphic as complex linear bundles over surfaces. 
Therefore, the celestial transform withstands exchange $\sky{}$ with $\Conjsky{}$. 
\end{remark}

%
%
\begin{definition}
$\Rfuture{}$ denotes the subbundle of non-negative functions in $\RC{}$. 
In a formula: $\Rfuture{} = \{\,\zeta \cdot \Conj\zeta\ |\ \zeta\in\cW{}\,\}$. 
An $\R$-linear bundle of real-valued (1,1)-homogeneous functions, 
containing $\Rfuture{}$, will be denoted by $\RR{}$. 
\end{definition}
\begin{remark}
The fibers of~$\Rfuture{}$ are equivalent to 
the ray~$[0, +\infty)$ up to multiplication by a positive number.
It is also easy to see that all values of the transform~$\sM$ defined above 
are, in fact, sections of~$\RR{}$.
\end{remark}
\begin{remark}
The bundles have naturally defined operations 
$\zeta\in\cW{} \mapsto \zeta \cdot \Conj\zeta $, denoted by ${\abs\ }^2$ (modulus squared) 
and $\abs\ : (T\sky{} = \mathcal{O}(2,0)) \to \Rfuture{}$, 
and both are continuous maps of bundles.
\end{remark}
The $\Rfuture{}$ is called the {\em bundle of sizes}, and its sections are {\em fields of sizes}.
The bundle/sections of~$\RR{} $ are {\em signed bundle/fields of sizes}.
\begin{remark} 
Representations of~$\mathrm{SL}(2,\C)$ on $\Csky{}$ and the spaces of spinors
generate a (consistent) action of the same group on the sky
that corresponds to the 
covering
$$
\{\pm 1\} \hookrightarrow \mathrm{SL}(2,\C) \ ^{2:1}\hspace{-1.5em}\longrightarrow 
\mathrm{PSL}(2,\C) \to 1.
$$
It preserves $T\sky{}$ and $\cW{}$ as holomorphic bundles 
($\RC{}$ and other~– as line bundles), 
while vector fields and sections of $\RC{}$ 
are transformed as they would be simply by the Riemann sphere’s motions.
\end{remark}

\begin{remark}
In accordance with the definition of~$\contactfib$ from~1.1, 
for each $w\in\sky{}\Mink,\ x := \xfib(w)$ we have to choose 
such null $v\in\mathcal{C}^+_x \noZero$ 
that represents the point $w\in\sky x$. 
Let’s express $v$ in the covariant form (i.e. to make it an element of~$T^*\Mink$) 
as $v_{AA'} = \scoord_{A} \mul \Conj\scoord_{A'}$. 
For each tangent vector~$\Xi\in T_w\sky{}\Mink$ 
we have $\contactfib(\Xi) = v \mul \xi = \scoord_{A} \mul \xi^{AA'} \mul \Conj \scoord_{A'},\ \xi := \xfib_*(\Xi)$, 
hence
\begin{equation}
\contactfib\quad\propto\quad
\scoord_{A}\mul d\xfib^{AA'}\mul\Conj\scoord_{A'}\ :\ T(\sky{}\Mink)\to\RR{}.
\tag{1}
\end{equation}
The symbol “$\propto$” here denotes equality up to multiplication by a positive number.
\end{remark}

\subsection{The construction of Poincaré-invariant conformal reference frame}
\begin{proposition} 
The collection $(\skyreg{} := \sky{}\Mink,\ M := \RR{},\ \mfib := \sM)$
is a conformal reference frame for~$\Mink$. 
\end{proposition}
\begin{remark}
The total space of the bundle $\RR{}$ can be represented as $\sky{} \times \R$, 
and any $\msky x$ is the graph of a section $\HB_A \cdot x^{AA'} \cdot \Conj\HB_{A'}$.
\end{remark}
\begin{proof}
Obviously, a map sending skies to graphs over~$\sky{}$
satisfies the condition (d) of the definition of a conformal reference frame. 
The condition (f) is satisfied because 
$\sM(\sky v)$, where $v$ is a null vector, 
is tangent to the zero section at the point $\mathrm{P}v$. 
To prove (c) in an easy way, let’s investigate all $w\in\sky 0\Mink$ 
which doesn’t lose generality because a translation of~$\Mink$ 
corresponds to the fiberwise addition (of a fixed smooth section) in the total space $\RR{}$. 
By virtue of the triviality of~$\sky{}\Mink$ mentioned in Sec.~2.1, we have $T_w\sky{}\Mink = T_w\sky 0 \oplus T_0\Mink$, 
where $T_w\sky 0 \subset \ker\contactfib$ and, 
moreover, $\mfib_*(T_w\sky 0)$ consists of vectors tangent to~$\msky 0$ (the zero section of~$\RR{}$). 
We choose, as in (1) at the end of the preceding subsection, 
such $\scoord\in\Csky{} \noZero$ that $\mathrm{P}\scoord = w$. 
A direct calculation
$$
\contactfib(\Xi)\ \ =\ \ \scoord_{A}\mul\xi^{AA'}\mul\Conj\scoord_{A'}
\ \ =\ \ \frac{\partial\mfib}{\partial x}|_\scoord (\xi),\quad\text{where~} 
\xi := \xfib_*(\Xi),\ \Xi\in T_w\sky{}\Mink
$$
shows that the sought image of the form $\contactfib| _w$ by $\mfib$ 
is simply the projection of the space~$T_{(\mathrm{P}\scoord, 0)}\RR{}$ onto 
the vertical direction (tangent to the fibers of $\RR{}$). 
The image of the form $\contactfib$ at an arbitrary point $w\in\sky{}\Mink,\ x := \xfib(w)$ 
would differ only in the projection that would go along the tangent 
to the graph $\HB_A \cdot x^{AA'} \cdot \Conj\HB_{A'} = \sM x$ 
instead of the one to the zero section. 
\end{proof}
A direct interpretation of the constructed reference frame is that 
the total space $\RR{}$ is the space of null hyperplanes in~$\Mink$, 
and $\toM$-preimage of an element of~$\RR{}$ is given by
a constant null direction—an element of~$\sky{}$—and a hyperplane in~$\Mink$.
Choosing $\mathrm{P}\scoord\in\sky{}$ and $\chi\in\RR{}|_{\mathrm{P}\scoord}$ arbitrarily, 
we have the set~$\xfib(\mfib^{-1}(\chi))$ of all points $\Mink$, 
$\sM$-images of whose skies pass through $\chi$, 
to become a hyperplane~– solution of
the linear equation $\scoord_A \mul x^{AA'} \mul \Conj\scoord_{A'} = \chi(\scoord)$ at $x^{AA'}$. 
Presence of the natural action of the Poincaré group on~$\RR{}$ ensures 
that the set of {\em all} null hyperplanes in~$\Mink$ can be obtained in such a way.
These null hyperplanes can also be understood 
as the light cones of “points at infinity” 
and the very three-dimensional total space $\RR{}$~– 
as the piece $\mathcal{I}$ of the {\em conformal compactification} $\widehat\Mink$ of the Minkowski space,\footnote{
	The conformal compactification~$\widehat\Mink$ is explained in \cite{PenroseRindler} Chap.~9,
	\cite{Jadczyk} 2.1, 2.2,
	and \cite{Vladimirov}.
	A more general construction of the conformal boundary of a Lorentzian manifold
	is given in \cite{Hawking}~6.8.
} 
which brings the present example closer to the one given in Sec.~1.2 
and will be justified in the next subsection.

\subsection{Interpretation through the twistor correspondence}
The manifold of null lines in the space $\widehat\Mink$ 
admits a known description as the projective null twistor space
$\mathbf{PN} \subset \mathbf{P}(\CConjsky{}\times\spinorsSections{})\simeq\C\mathbf{P}^3$. 
$\mathbf{N}$ is given by $\Conj\pi_L\mul\omega^L + \pi_{L'}\mul\Conj\omega^{L'} = 0\,$\cite{TwistorProgramme} 
and it’s 
a smooth real hypersurface in~$\CConjsky{} \times \spinorsSections{}$, 
and $\mathbf{PN}$ lies in the respective complex projective space. 
If we restrict the consideration to the Minkowski space proper ($\X := \Mink$), 
then $\mathfrak{N} = \mathbf{P}\Nfin \subset \mathbf{PN}$ (hereinafter called the {\em affine part} of $\mathbf{PN}$).
Moreover, $\twId{\sky x}$ is given by the following equation:
\begin{equation}
\omega^L = i\mul x^{LL'} \mul \pi_{L'} \quad
\text {\cite {TwistorProgramme},} 
\tag{2}
\end{equation}
where $\CConjsky{}\ni\pi = \Conj\scoord$ are complex projective coordinates 
in the fibers of $\Conjsky{}\Mink$. 
$\Nfin$ is the union of all solutions of equation~(2) for all $x\in\Mink$, 
and $\Nfin = \mathbf{N} \setminus \{\pi = 0\}$. 
The smoothness of the mapping $\twId\,: \sky{}\Mink \to \mathbf{P}\Nfin$ is obvious. 
\begin{definition}
The {\em twistor contraction} $\tau: \mathbf{P}(\CConjsky{} \noZero\ \times \spinorsSections{}) \to \RC{}$ 
maps a point~$\mathrm{P}(\pi,\omega)$ to an element of the total space of the bundle~$\RC{}$ over $\Conjsky{}$ 
(identical to $\RC{}$ over $\sky{}$),
expressed with the point $\mathrm{P}\pi$ on the base  
and the “complex size” $\lambda\mul\pi \mapsto -i\mul\abs\lambda^2\mul\Conj\pi_L\mul\omega^L$.
\end{definition}
\begin{remark}
Well-definedness of the of the twistor contraction 
requires that specified (1,1)-homogeneous function on~$\{\lambda\mul\pi\}$, 
depending on the pair~$(\pi,\omega)$ as on a parameter, 
may not change when projective coordinates multiplied by a non-zero complex number. 
This requirement is obviously satisfied.
\end{remark}
\begin{remark}
The definitions of $\tau$ and $\mathbf{N}$, as well as the explicit form of $\Nfin$ imply that
$$
\Nfin\ \ =\ \ \{(\pi\in\CConjsky{}\noZero,\ \omega\in\spinorsSections{})\ |\ \tau(\pi,\omega)\in\RR{}\,\}
$$
(in other words, $\mathbf{P}\Nfin = \tau^{-1}(\RR{})$). 
Obviously, the $\tau$-preimage of every element of~$\RR{}$ is the affine part 
of a projective line in~$\mathbf{PN}$ (or a plane in~$\mathbf{N}$), 
but with a {\em constant} $\mathrm{P}\pi$ unlike those defined by the equation~(2).
\end{remark}
\begin{remark}
The map~$\toM$ introduced in the previous subsection coincides with the composition 
of the projection $\twId\,: \sky{}\Mink \to \mathbf{P}\Nfin$ with the twistor contraction $\tau$;
it’s demonstrated by contracting (2) with $-i\mul \Conj\pi_L\,$. 
Hence, $\tau$-preimage of each element $\chi\in\RR{}$ 
specifies a point at infinity in~$\widehat\Mink$, 
where null lines comprising 
the null hyperplane $\xfib(\mfib^{-1}(\chi))$ 
intersect.
\end{remark}
%
\begin{remark}
From 
the expression of $\mfib$ in terms of $\tau$ 
and the fact that the projective line~$\omega = 0$ is $\twId{\sky 0} \subset \mathbf{P}\Nfin$, 
we find that, similarly to the Proposition~1, 
the contact form on~$\mathbf{P}\Nfin$, restricted to~$\twId{\sky 0}$, 
admits the expression 
$$
\contact|_{\twId{\sky 0}}\ \ =\ \ d\tau|_{\{\,\omega\,=\,0\,\}}
\ \ =\ \ -i\mul\Conj\pi_L\mul d\omega^L|_{\{\,\omega\,=\,0\,\}}\,.
$$
In view of the homogeneity of $\tau$, the expression $d\tau$ 
consistently specifies 
a mapping $T\mathbf{P}\Nfin|_{\twId{\sky 0}} \to \RC{}$.

It is also easy to show that on the whole $\Nfin$:
$$
\contact_{\Nfin} 
\quad=\quad -i\mul(\Conj\pi_L\mul d\omega^L + \Conj\omega^{L'}\mul d\pi_{L'}).
$$
\end{remark}

\section{The flow of time equation}
In the general Lorentzian case we have to consider 
line bundles over {\em different} skies. 
Let’s $\cW x$, $\Rfuture x$, and so on
denote respective bundles over the sky $\sky x$. 
Similarly (with subscript “$x$”) the spaces of spinors, depending on~$x$, will be indicated.
\begin{remark}
It is easy to see that the inverse image $\xfib^*(T_x\X) = \Th\sky x$ of the tangent space 
is $\sky x \times T_x\X$ that turns to~$\sky{}(T_x\X)$ 
after the base is exchanged with the fiber. 
Moreover, $\sky{}(T_x\X) $ is isomorphic to~$\sky{}\Mink$ (considered in Sec.~2), 
and the isomorphism is defined up to the action of the Lorentz group.
\end{remark}
\begin{definition}
Let’s denote by $\RR{}\X$ 
the disjoint union of the signed bundles of sizes $\RR x$ for all $x\in\X$, 
which gives a smooth line bundle over~$\sky{}\X$. 
Similarly,, $\Rfuture{}\X \subset \RR{}\X $ is the disjoint union of the bundles of sizes $\Rfuture x$ for all $x\in\X$.
\end{definition}
\begin{definition}
The map $\s: \Th\sky{}\X \to \RR{}\X$ is defined 
fiberwise in terms of $\s_{T_x\X}: \Th\sky x \to \RR x$ for all $x\in\X$,
where $\s_{T_x\X}$ is identical to the map~$\sM$ introduced in Sec.~2.2. 
\end{definition}
\begin{remark}
For a chosen (arbitrarily) an orientation on~$T_x\X$, 
well-definedness of~$\s_{T_x\X}$ 
follows from the $\mathrm{SL}(2,\C)$-invariance of the space~$\spinorsSections x$ 
and the weights of the respective representation of the group in~$\RC x$. 
Reversing the orientation on~$T_x\X$ interchanges 
$\sky x$ and $\HB_A \in\spinorsSections x$ with $\Conjsky x$ and $\ConjHB_{A'} \in \ConjspinorsSections x$, respectively, 
but $\RR x$ and the mapping of the bundle $\Th\sky x$ to it don’t change.
\end{remark}
Now we are able to supersede the definition of the form~ $\contactfib$, given in~1.1, 
with the\vspace{-1.25ex}
\begin{definition}
$ \contactfib := \s \circ d\xfib$ is an 1-form on the total space $\sky{}\X $ with values in~$\RR{}\X$, 
defined by the composition of the horizontal projection $d\xfib: T(\sky{}\X) \to \Th \sky{} \X$ 
with the map~$\s$ introduced above. 
\end{definition}
\begin{remark}
The “new” form~$\contactfib: T\skyreg{} \to \RR{} \X$ 
and the $\R$-valued form defined in Sec.~1.1 up to a positive factor, 
coincide as distributions of oriented co-directions on~$\sky{}\X$.
\end{remark}
\begin{remark}
It follows directly from the definition of $\contactfib$ 
that the celestial transform is
expressed as 
$\s = d\xfib\,_*(\contactfib) $. 
\end{remark}
\begin{theorem}
For any conformal reference frame~$(\skyreg{}, M, \toM)$ for a Lorentzian manifold~$\X$ 
the projection $\bPm$ (see~1.3) satisfies
$$
\ker \bPm
\quad = \quad
\ker \contactfib|_{\skyreg{}}\,.
$$
\end{theorem}
\begin{proof}
The kernel of~$d\mfib$ in~$T\skyreg{}$, according to the definition of the conformal reference frame, is three-dimensional, 
lies entirely in~$\ker\contactfib$, 
and has the zero intersection with $\Tv\skyreg{}$,
hence $\ker\contactfib|_{\skyreg{}} = {\ker d\mfib}\, \oplus \,{\Tv\skyreg{}}$.
Because 
the fiberwise normal bundle $\NM\skyreg x$—the range of~$\bPm $—is obtained 
from~$\mfib|_{\skyreg x}^*(TM)$ as the quotient by the image of~$d\mfib(T\skyreg x)$, 
the theorem is proved.
\end{proof}
\begin{corollary}
The map $d\msky x$ (derivative of the sky image with respect to~$x$) 
nullifies $d\xfib(\ker\contactfib)$, and only this. 
See the diagram at the end of~1.3. 
\end{corollary}
\begin{corollary}
The direct image of the 1-form~$\contactfib$ under~$\bPm$ 
is well-defined, and doesn’t vanish anywhere on~$\skyreg{}$, 
giving an isomorphism of 
the bundles $\RR{}\skyreg{} := \RR{}\X|_{\skyreg{}}$ and $\NM\skyreg{}$.
\end{corollary}
\begin{theorem}
Let 
$\X$ be an oriented manifold with a conformal reference frame, 
satisfying the conditions in Theorem~1. 
We define, based on the last Corollary, the homomorphism of line 
bundles $a_{\toM} := \bPm_*(\contactfib)^{-1} : \RR{}\skyreg{} \to \NM\skyreg{}$.
\footnote{
	In other words, $a_{\toM} = \contactfib_* (\bPm)$ and the solution of the equation $a_{\toM} \circ \contactfib = \bPm $.
} 
Then the derivative (see Sec.~1.3) 
of the family $\{\msky x\} = \{\mfib(\skyreg x)\}$ satisfies the so-called “flow of time equation”:
$$
d_{AA'}\,\msky x = a_{\toM}\mul(\HB_A\cdot\ConjHB_{A'})
\quad\Leftrightarrow\quad
\bPm_*(\contactfib)\cdot d_{AA'}\,\msky x = \HB_A\cdot\ConjHB_{A'}\,,$$
where $\HB_A $ and $\ConjHB_{A'}$ are 
bases in $\spinorsSections x$ and $\ConjspinorsSections x$ complex conjugate to each other 
(or, equivalently, the coordinates in~$\Csky{x\,}$ and $\CConjsky{x\,}$), respectively, 
in which the algebraic spinor correspondence for the differential~$d_{AA'}\,\msky x$ is expressed.
\end{theorem}

\begin{proof}
From the identity $\bPm = d\msky x \circ d\xfib $ presented 
in the diagram in Sec.~1.3 
and the definition of~$\contactfib$, we deduce that $\bPm_*(\contactfib) = (d\msky x)_*(\s)$. 
And $\s\xi = \HB_A \cdot \xi^{AA'} \cdot \ConjHB_{A'}$ by construction, where $\xi\in T_x\X$, 
which proves the theorem.
\end{proof}

\section{Other dimensions of space}
As a generalization of Lorentzian manifolds, we can consider 
pseudo-Riemann manifolds of the signature~(1,{\it d}), i.e. one plus and {\it d} minuses 
(where {\it d} is a natural number), and the time orientation at each point. 
Section~1 is applicable to this case with obvious changes on the dimension: 
a sky becomes $S^{d-1}$, $\Delta$ becomes the $d\!-\!1$-dimensional ball, and so on.
The complex algebra in Section~2 does not admit an obvious generalization, 
because only the case (1,3) has the spinor group isomorphic to a complex special linear group. 
Nevertheless, the space $\Md$ admits the conformal compactification, 
and the invariant frame of reference built in Sec.~2.3, 
making it possible to define the bundle of sizes $\RR{}\X$ as the range of the contact form. 
Theorem~1 also holds for any {\it d}.

\section{Appendix – The causality} 
This part of the paper isn’t included to the peer-reviewed text going to publication. 
Note that this research was made in early 2016 and 
{\bf some ideas expressed here do not agree 
with the recent (as of May, 2017) author’s understanding of the causality.}

The causal relation in Relativity is understood in terms of Lorentzian manifolds. 
This is fairly convenient in Special Relativity, but 
has several shortcomings for arbitrary Lorentzian space-times, 
and in General Relativity when geometry (as the metric tensor) 
effectively becomes one of the field variables.
The manifold substantionalism leads to causal relation 
expressed in terms of causal paths and sets
that can be seen in \cite{LowConjecture} and elsewhere,
but it’s presently unknown whether this concept matches the fundamental physical one.

The aim of this section is to propose a construction of the causality
consistent with space-time that is a Lorentzian manifold~$\X$,
but based on other, more quantum-friendly foundations. 
This construction of the causality will somewhere differ from predictions of General Relativity. 
More precisely, let’s assume $\X$ as some approximation to the physical world.
Now we want to:
\begin{itemize} 
\item Give a description of the causal relation in terms external to~$\X$. 
\item Propose a kind of absolute reference frame suitable for our universe. 
\item Permit for interpretation of Quantum Mechanics along the lines of \cite{PresheafQR}.
\end{itemize}
As defined in Sec.~1.2, there may be {\em several} (essentially) different 
conformal reference frames for the same $\X$. 
For Minkowski space, we obviously identify our $\RR{}$ with $\mathcal{I} = \mathcal{I}^+ = \mathcal{I}^-$,
whereas for another similar $\X$, that is only asymptotically Minkowski space, 
projections to the future null infinity $\mathcal{I}^+$ and to the past null infinity $\mathcal{I}^-$
will result in different reference frames.
Does the current physical cosmology propose any hint about the natural conformal reference frame?
Fortunately, the “cosmological censorship” principle implies 
that any null geodesic extended to the past must meet the Cosmological Singularity
of the FLRW cosmology.
Since the latter is 3-dimensional conformal boundary of our space-time (see e.g. \cite{Hawking}),
it’s an obvious candidate for~$M$.
Let $M$ be the Cosmological Singularity from now on 
and let $\toM$ denote the natural operation of geodesic continuation towards the Singularity 
and taking the intersection point then, similarly to the Example from Sec.~1.2.
This conformal reference frame will be called the {\em absolute reference frame}.

For such a universe—without enough curvature to make gravitational lensing—%
as a pure Robertson–Walker space,
the picture will not differ greatly 
from the one explained in Sec.~2.3 for the Minkowski space.
$\skyreg{}$ should be taken as the whole $\sky{}\X$, 
and all sky images will be smooth spheres embedded to~$M$.
Causality relation will admit a simple geometric description.
Let $\mathfrak{P}_x$ denote the union of~$\msky x$ with its interior.
Then any event belongs to the causal past of~$x$ if and only if its sky image lies in~$\mathfrak{P}_x$.
This can be formulated for the Minkowski case as well (albeit with little merit),
replacing “interior” with “the negative side in~$\RR{}$”.

For a more realistic cosmology $\mfib_*\,|_{\Tv\skyreg{}}$ will degenerate somewhere.
Similarly to the Example, $\skyreg{}$ can be defined as the open set where $\toM$ doesn’t
degenerate. 
But we can also consider the {\em full sky image} of~$x$, 
since $\mfib$ is defined everywhere on~$\sky{}\X$ and is continous.
It is a topologically closed surface in~$M$ despite singularities (but not necessarily a submanifold).
Then, $\mathfrak{P}_x$ can be taken as a 3-dimensional subset of~$M$ as well: 
the union of the full sky image of~$x$ with its interior.
A redefined “causal relation” depending of $\mathfrak{P}_x$ 
may be 
weaker that the one of General Relativity even in some globally hyperbolic cases.
Moreover, a realistic model of the universe can’t be globally hyperbolic.
Particularly, the (absolute reference frame)-based causality relation 
abolishes the event horizon in a black hole (which could shed light to the “information paradox”).

Connecting this to the “Locale of Time” concept requested in~\cite{PresheafQR}~1.3, 
a possible construction is the space of {\large closed subsets of~$M$}.
It’s a bounded join-semilattice (with the set union “$\cup$” as the join)
and has a natural non-Hausdorff topology specified with the base
$$
\left\{\,
	\left\{\,B\text{~is a closed subset of~}M\ |\ B\cap K = \emptyset\,\right\}\ |
	\ K\Subset M\,
\right\}.
$$
Substituting $K := \mathfrak{P}_x$ for any space-time event~$x$, 
we have an open set in the Locale of Time necessary for 
definition of open balls in the Space of Ultimations discussed in that paper.
\vspace{4ex}\strut

\pagebreak[3]

\end{document}